\documentclass[11pt]{article}
\usepackage[utf8]{inputenc}

\usepackage{graphicx}% Include figure files
\usepackage{xcolor}
\usepackage{amsmath,amsthm,amssymb}
\usepackage[margin=1in]{geometry}
\usepackage[pdftex,colorlinks=true,linkcolor=blue,citecolor=blue,urlcolor=black]{hyperref}
\usepackage{comment}
\usepackage[normalem]{ulem}
\usepackage{cite}
\usepackage{authblk}

\newcommand{\F}{\mathbb{F}}

\def\e{{\bf e}}
\def\r{{\bf r}}
\def\x{{\bf x}}
\def\y{{\bf y}}
\def\w{{\bf w}}

\newcommand{\ket}[1]{| #1 \rangle}

\newcommand{\braket}[1]{\langle #1 \rangle}

\DeclareMathOperator{\Inf}{Inf}

\newcommand{\fgor}{f_{G,\operatorname{OR}}}
\newcommand{\fgpar}{f_{G,\operatorname{Par}}}

\newcommand{\be}{\begin{equation}}
\newcommand{\ee}{\end{equation}}
\newcommand{\bea}{\begin{eqnarray}}
\newcommand{\eea}{\end{eqnarray}}
\newcommand{\bes}{\begin{equation*}}
\newcommand{\ees}{\end{equation*}}
\newcommand{\beas}{\begin{eqnarray*}}
\newcommand{\eeas}{\end{eqnarray*}}

\makeatletter
\newtheorem*{rep@theorem}{\rep@title}
\newcommand{\newreptheorem}[2]{%
\newenvironment{rep#1}[1]{%
 \def\rep@title{#2 \ref{##1} (restated)}%
 \begin{rep@theorem}}%
 {\end{rep@theorem}}}
\makeatother

\allowdisplaybreaks

\newtheorem{thm}{Theorem}
\newtheorem*{thm*}{Theorem}

\newtheorem{lem}[thm]{Lemma}
\newtheorem*{lem*}{Lemma}

\newtheorem{prop}[thm]{Proposition}

\newtheorem{fact}[thm]{Fact}

\newreptheorem{thm}{Theorem}
\newreptheorem{lem}{Lemma}
\newreptheorem{cor}{Corollary}

\title{Quantum algorithms for learning a hidden graph and beyond}

\author[1,2]{Ashley Montanaro\thanks{ashley.montanaro@bristol.ac.uk}}
\author[1]{Changpeng Shao\thanks{changpeng.shao@bristol.ac.uk}}
\affil[1]{School of Mathematics, University of Bristol, UK}
\affil[2]{Phasecraft Ltd.}
\begin{document}

\maketitle

%-------------------------------------------------------------------------------------------------------

\begin{abstract}
We study the problem of learning an unknown graph provided via an oracle using a quantum algorithm. We consider three query models. In the first model (``OR queries''), the oracle returns whether a given subset of the vertices contains any edges. In the second (``parity queries''), the oracle returns the parity of the number of edges in a subset. In the third model, we are given copies of the graph state corresponding to the graph.

We give quantum algorithms that achieve speedups over the best possible classical algorithms in the OR and parity query models, for some families of graphs, and give quantum algorithms in the graph state model whose complexity is similar to the parity query model. For some parameter regimes, the speedups can be exponential in the parity query model. On the other hand, without any promise on the graph, no speedup is possible in the OR query model.

A main technique we use is the quantum algorithm for solving the combinatorial group testing problem, for which a query-efficient quantum algorithm was given by Belovs. Here we additionally give a time-efficient quantum algorithm for this problem, based on the algorithm of Ambainis et al.\ for a ``gapped" version of the group testing problem. We also give simple time-efficient quantum algorithms  based on Fourier sampling and amplitude amplification for learning the exact-half and majority functions, which almost match the optimal complexity of Belovs' algorithms.

% Finally, we give query-efficient quantum algorithms and quantum lower bounds for the problem of learning monotone disjunctive normal form boolean formulae, which is a generalization of learning graphs with OR queries and combinatorial group testing.
%Belovs' algorithm is based on adversary bound method, which is beautiful but rather complex, and leads to algorithms which are not necessarily efficient in terms of time complexity. Here we give a quantum algorithm with efficient implementations based on [Ambainis et al. arXiv:1507.03126] for a ``gapped" version of the group testing problem. We also give simple time efficient quantum algorithms  based on Fourier sampling and amplitude amplification for learning exact-half and majority functions, which match the optimal complexity of Belovs' algorithms. 
%The ideas we used may be of independent interest and may have other applications.

%The learning problem of $s$-term $r$ monotone DNFs is a generalization of learning graphs with OR queries and combinatorial group testing. Here we apply Grover's algorithm to accelerate the currently best known classical algorithms given in [Abasi et al. arXiv:1405.0792]. 
%A quantum lower bound is also given based on the optimality of Grover's algorithm. Similar to the classical algorithms, the quantum 
%algorithms are asymptotically tight for fixed $r$ and $s$.
\end{abstract}

%-------------------------------------------------------------------------------------------------------

\section{Introduction}

Quantum computers are known to be able to compute certain functions more quickly than their classical counterparts, in terms of the number of queries to the input that are required. In some cases, quantum algorithms can also learn unknown objects using fewer queries than their classical counterparts. 
%For instance, if the function on $n$-bits is assumed to be either constant or balanced, then Deutsch-Jozsa algorithm uses 1 query \cite{deutsch1992rapid}. While any deterministic classical algorithm needs $2^{n-1}+1$ queries.
For example, if we are given query access to an unknown boolean function on $n$-bits which is promised 
to be a dot product between $x$ and a secret string $s$ modulo 2, then the Bernstein-Vazirani algorithm learns this function with 1 query \cite{bernstein97}, while the best possible classical algorithm uses $n$ queries. If the function is promised to be an OR function of $k$ unknown variables, then Belovs' algorithm for combinatorial group testing~\cite{belovs15} learns this function with $\Theta(\sqrt{k})$ queries, while the best possible classical algorithm needs $\Theta(k\log (n/k))$ queries. These speedups are not far from the largest quantum speedups that can be achieved. For any class $\mathcal{C}$ of Boolean functions over $\{0,1\}^n$, let $D$ and $Q$ be such that an unknown function from $\mathcal{C}$ can be identified using $D$ classical membership queries or from $Q$ quantum membership queries. 
Then $D = O(nQ^3)$~\cite{servedio2004equivalences}.

%In, Servedio and Gortler
%showed that in an information-theoretic sense, quantum and classical learning are equivalent up to polynomial factors: for the model of exact learning from membership queries, there is no learning problem which can be solved using significantly fewer quantum queries than classical queries.  More precisely, Let $\mathcal{C}$ be any class of Boolean functions over $\{0,1\}^n$, and let $D$ and $Q$ be such that $\mathcal{C}$ is exact learnable from $D$ classical membership queries or from $Q$ quantum membership queries. Then $D = O(nQ^3).$ 
%\am{Insert references to quantum algorithms and lower bounds for learning -- Servedio / Gortler?}

Here we focus on the problem of learning an unknown graph using quantum queries, in a variety of settings. Many quantum speedups (both polynomial, e.g.~\cite{durr06}, and exponential, e.g.~\cite{ben2020symmetries}) are known for problems involving graphs. However, the only quantum speedup we are aware of for \emph{learning} graphs is recent work on learning graphs using cut queries~\cite{lee20}.

We consider several different notions of queries to an unknown graph -- OR queries, parity queries, and graph states, all defined below -- and aim to minimize the number of queries required to identify the graph. The first two of these query models are closely related to models that have been extensively studied in the classical literature on exact learning, e.g.\cite{angluin08,grebinski2000optimal,choi2010optimal}, in particular because of their applications to computational biology. In some cases we find polynomial speedups over the best possible classical complexity, while in other cases (such as learning bounded-degree graphs in the parity query model) the speedups can even be exponential.

A summary of our results is as follows; also see Table \ref{tab:graphs}. Throughout, we use $n$ to denote the number of vertices and $m$ to denote the number of edges of a graph.

\begin{table}[t]
    \[
    \begin{array}{|c|c|c|c|c|c|}
    \hline
%& \multicolumn{2}{c|} \vee & \multicolumn{2}{c|} \oplus & \ket{G}\\ \hline 
%
& \text{Q, } \vee & \text{Q, } \oplus & \text{C, } \vee,\oplus & \text{Q}, \ket{G} \\ \hline 
\text{All graphs} & \Theta(n^2) & \Theta(n) & \Theta(n^2) & \Theta(n) \\
\text{$m$ edges} & O(m \log (\sqrt{m} \log n) + \sqrt{m} \log n) & O(\sqrt{m\log m}) & \Omega(m \log \frac{n^2}{m}) & O(m \log \frac{n^2}{m})\\
%
%\text{Bipartite} & O(m \log (\sqrt{m} \log n) + \sqrt{m} \log n)  & O(\log m) & \Omega(m\log \frac{n}{m}) & O(\log m) \\
%
\text{Degree $d$} & O(nd \log (\sqrt{dn} \log n) + \sqrt{nd} \log n)  & O(d \log \frac{m}{d}) & \Omega(nd\log\frac{n}{d}) & O(d \log m) \\
\text{Matching} & O( m^{3/4}\sqrt{(\log n)}(\log m)  + \sqrt{m} \log n)  & O(\log m) & \Omega(m\log\frac{n}{m}) & O(\log m) \\
\text{Cycle} & O( m^{3/4}\sqrt{ (\log n)} (\log m)+ \sqrt{m} \log n)  & O(\log m) & \Omega(m \log \frac{n}{m}) & O(\log m) \\
\text{Star} & \Theta(\sqrt{m})  & O(1) & \Omega(m \log \frac{n}{m}) & O(1) \\
\text{$k$-vertex clique} & \Theta(\sqrt{k})  & O(1) & \Omega(k\log\frac{n}{k}) & O(1) \\ \hline
\end{array}
\]
\caption{Query complexities for learning various classes of graphs: $m$ is the number of edges, $n$ is the number of all vertices. The symbols $\vee$ (OR), $\oplus$ (parity), and $\ket{G}$ (graph state) denote the type of query considered. Q and C denote quantum and classical queries.}
\label{tab:graphs}
\end{table}

\begin{enumerate}
    \item \textbf{(OR queries)} First, we consider the problem of identifying an unknown graph, given access to queries to subsets of the vertices, which return whether the corresponding induced subgraph has any edges within that subset. That is, given a graph $G = (V,E)$, a query takes a subset $S \subseteq V$ and returns whether $E \cap (S \times S)$ is empty. This model has been extensively studied classically and we will briefly survey these results below. Our main results in this model are:
    \begin{itemize}
        \item A quantum algorithm to learn an unknown graph with $m$ edges using $O(m \log ( \sqrt{m} \log n) $ $+ \sqrt{m} \log n)$ OR queries, as compared with the classical lower bound of $\Omega(m\log (n^2/m))$. For some relationships between $m$ and $n$ (e.g.\ $m=\Theta(\log n)$) this gives a modest quantum-classical separation.
        \item The lower bound that any quantum algorithm that identifies an arbitrary unknown graph in this model must make $\Omega(n^2)$ OR queries, so the above algorithm's complexity cannot be improved by more than log factors.
        \item Learning graphs with special structure, such as Hamiltonian cycles, matchings, stars and cliques, has specific applications in molecular biology  \cite{grebinski97,grebinski1998reconstructing,alon04}. We give quantum speedups for learning these graphs in this model. The graphs and quantum speedups can be roughly summarized as follows. Hamiltonian cycles and matchings: $k^{3/4}$ vs.\ $k$; stars and cliques: $\sqrt{k}$ vs.\ $k$. Here $k$ is the number of non-isolated vertices.
    \end{itemize}
    
    \item \textbf{(Parity queries)} Next, we consider the same problem, but where the oracle returns the \emph{parity} of $|E \cap (S \times S)|$, for arbitrary subsets $S$. Although this may seem a more unusual setting, this oracle can be obtained from the perhaps more natural oracle, known as additive oracle, that returns the size of $E \cap (S \times S)$, which has also been studied classically\cite{grebinski2000optimal,choi2010optimal,bouvel05,reyzin2007learning}. We will see that larger quantum speedups are available in this model.
    %Another motivation for us to study this oracle is its close connection to graph state.
    Here, we show that:
    \begin{itemize}
        \item There is a quantum algorithm which learns an unknown graph with degree $d$ making $O(d \log m)$ parity queries, as compared with the classical lower bound of $\Omega(nd \log(n/d))$ queries.
        \item There is a quantum algorithm which learns an unknown graph with $m$ edges making $O(\sqrt{m\log m})$ parity queries, as compared with the classical lower bound of $\Omega(m\log (n^2/m))$.
        \item Stars and cliques can be learned with $O(1)$ parity queries.
    \end{itemize}
    Our results show that, for some families of graphs, parity queries can be exponentially more efficient than OR queries for quantum algorithms. The results we obtain are based on very similar ideas to a recent work by Lee, Santha and Zhang~\cite{lee20}, which considered a related ``cut query'' model (see below).
    
    \item \textbf{(Graph states)} We also study a quantum version of the problem of learning an unknown graph: the problem of learning an unknown graph state~\cite{hein06}. Graph states are a family of quantum states that have many important applications, in particular to measurement-based quantum computing. Any graph $G$ has a corresponding graph state $\ket{G}$, and it is a natural question to ask how many copies of $\ket{G}$ are required to identify $G$. It was already known that $\Theta(n)$ copies are necessary and sufficient if $G$ is an arbitrary graph with $n$ vertices~\cite{aaronson08,zhao16,montanaro17a}. However, we show that one can do better given some additional information about $G$:
    \begin{itemize}
        \item If $G$ has degree $d$, we can learn $G$ using $O(d \log m)$ copies. If $G$ is promised to be a subgraph of a known graph $G'$ with bounded degree $d$, the quantum algorithm is also time-efficient (has runtime $\widetilde{O}(d^3 n)$). This second algorithm could be particularly useful in the practically-relevant scenario where we aim to produce a desired graph state $G'$, but some edges of $G'$ have failed to be generated, and we would like to determine which edges have failed.
        \item If $G$ is known to be picked from a set of size $L$, we can learn $G$ using $O(\log L)$ copies. For example, if $G$ is known to have at most $m$ edges, we can learn $G$ using $O(m \log n)$ copies.
    \end{itemize}
    The results about learning graph states also underpin the results about learning graphs from parity queries, because it turns out that using a procedure known as Bell sampling~\cite{montanaro17a} to learn a graph state is equivalent to learning a graph using parity queries -- except with the restriction that these queries are only to uniformly random subsets of the vertices.
\end{enumerate}

We also find improvements to quantum learning algorithms for other types of boolean functions.
Belovs~\cite{belovs15} gave optimal quantum query algorithms for learning an unknown boolean function $f:\{0,1\}^n \to \{0,1\}$ such that $f$ is promised to depend on a subset $S$ of the input bits of size $k$, and to correspond to a known function $g:\{0,1\}^k \to \{0,1\}$ on those bits. Thus the task is to identify $S$ by evaluating $f$ on different inputs. Belovs gave quantum algorithms using $O(\sqrt{k})$ queries when $g$ is the OR function, and $O(k^{1/4})$ queries when $g$ is the exact-half or majority function. The case of the OR function is particularly natural, because this corresponds to the well-studied problem of combinatorial group testing (CGT, also known as ``pooled testing'')  \cite{du2000combinatorial,porat2008explicit}, and we use this algorithm extensively as a subroutine for learning graphs.

Belovs' algorithms are produced by directly solving the semidefinite program for the general adversary bound, which is known to characterise quantum query complexity. This approach is beautiful but rather complex, and leads to algorithms which are not necessarily efficient in terms of time complexity. Here:
    \begin{itemize}
        \item We give a quantum algorithm for the case of the OR function that makes $\widetilde{O}(\sqrt{k})$ queries and runs in time $\widetilde{O}(n\sqrt{k})$, based on the use of an algorithm of Ambainis et al.~\cite{ambainis16} for a ``gapped'' version of the group testing problem.
        \item We give simple explicit algorithms for the cases of the exact-half and majority functions which match the complexity of Belovs' algorithms. The algorithms are based on Fourier sampling combined with amplitude amplification. We observe that a similar approach can be applied to many other boolean functions.
    \end{itemize}

\subsection{Summary of the techniques}

{\bf The OR query model.} In this model, we use a similar strategy to the classical algorithm given by Angluin and Chen \cite{angluin08}. The basic idea of \cite{angluin08} is binary search: We decompose the set of vertices $V$ into halves $V_1,V_2$, and suppose we already know the edges in $V_1, V_2$. We then try to learn the edges between them. The edges in $V_1, V_2$ can be learned recursively, and the complexity is dominated by the learning of the edges between $V_1, V_2$. This is an adaptive algorithm.
On a quantum computer, we can use the quantum algorithm for CGT to accelerate the learning of the edges between $V_1, V_2$. However, the classical inductive idea may not be applicable to the quantum case. A reason is that the underlying constant in the complexity of quantum algorithm for CGT is unknown for us, so we cannot bound the overall complexity easily. To overcome this problem, we first decompose $V$ into a disjoint union of some subsets such that each subset contains no edges, then learn the edges between the subsets. This idea is inspired by the non-adaptive learning algorithm of \cite{angluin08}.

{\bf The graph state model.} In this model, we apply Bell sampling \cite{montanaro17a} to learn an unknown graph state. Each Bell sampling returns  a uniformly random stabilizer of the graph state. Equivalently, if $A$ is the adjacency matrix of the graph, then each Bell sample returns $A \mathbf{s}~(\rm mod~2)$ for  a random vector $\mathbf{s} \in \{0,1\}^n$. If we take $k$ samples, then we obtain an $n\times k$ matrix $B$ and the matrix $AB$. From $B$ and $AB$ we can uniquely determine $A$ by choosing a suitable $k$.

{\bf The parity query model.} Since the graph state can be generated by
a parity query on a uniform superposition, any results for the graph state model also hold for the parity query model. Differently from the graph state model, with parity queries, we do have control of $\mathbf{s}$. More precisely, for any $\mathbf{s} \in \{0,1\}^n$, there is a quantum algorithm that returns $A \mathbf{s}~(\rm mod~2)$ using two parity queries. With this result, we can learn graphs of $m$ edges more efficiently by considering the low and high-degree parts.

%-------------------------------------------------------------------------------

\subsection{Prior work}

\textbf{Learning graphs with OR queries.}
Graph learning appears in many different contexts. In different applications, we apply different queries, and the OR query is important for problems in computational biology.
This type of query is also known as independent set query \cite{beame2020edge} and edge-detection query \cite{angluin08}. Many classical algorithms were discovered to learn graphs using OR queries in the past decades. 
For special graphs,
Beigel et al.~\cite{beigel01} and Alon et al.~\cite{alon04} have given algorithms for learning an unknown matching using $O(n \log n)$ queries.
Grebinski and Kucherov~\cite{grebinski97} gave an algorithm for learning a Hamiltonian cycle using $O(n \log n)$ queries.
Alon and Asodi~\cite{alon05} gave bounds on nonadaptive deterministic algorithms for learning stars and cliques. 
Bouvel et al.~\cite{bouvel05} gave algorithms for learning an unknown star or clique using $O(n)$ queries. The constant factors in the algorithms for learning Hamiltonian cycles, matchings, stars and cliques were improved by Chang et al.~\cite{chang11}. 

In the general case,
Angluin and Chen~\cite{angluin08} gave a deterministic adaptive algorithm with complexity $O(m \log n)$ for learning a graph with $m$ edges, encompassing all the above bounds (however, note that other restrictions can be considered, such as nonadaptivity, or restricted levels of adaptivity). The constant factor in this runtime was improved by Chang, Fu and Shih~\cite{chang14}.
The complexity $O(m\log n)$ obtained in~\cite{angluin08} assumes $m$ is known in advance. When $m$ is not known, the complexity of \cite{angluin08} is $O(m\log n + \sqrt{m}\log^2n)$. This is recently improved to $O(m\log n + \sqrt{m}(\log n) (\log  \overset{k}{\dots} \log n))$ in \cite{abasi2019learning}, where $k$ can be any constant. 

\textbf{Graph states.} In \cite{zhao16}, Zhao, P{\'e}rez-Delgado and Fitzsimons studied the problem of representing basic operations of graphs by graph states with high efficiency and showed that no classical data structure can have similar performance. In this work, the authors gave an algorithm for learning an arbitrary graph state of $n$ qubits using $O(n)$ copies. Graph states are a subclass of stabilizer states. Alternative algorithms for learning an arbitrary stabilizer state with $O(n)$ copies have been given by Aaronson and Gottesman~\cite{aaronson08} and Montanaro~\cite{montanaro17a}.

\textbf{Learning graphs with parity queries.} The parity query model is a special case of a model for graph queries which generalises the OR query model, and is known as additive queries \cite{grebinski2000optimal,choi2010optimal} (also known as quantitative queries \cite{bouvel05} and edge counting queries \cite{reyzin2007learning}).  The additive query plays an important role for applications related to DNA sequencing.
In this model, a query to a subset $S$ returns the number of edges of $G$ in $S$; the parity query model is obtained if this answer is taken mod 2.

The additive query is known to be somewhat more powerful than the OR query for learning graphs. For instance, as shown in \cite{bouvel05}, a Hamiltonian cycle or a matching can be identified with $O(n)$ additive queries, while this requires at least $\Omega(n\log n)$ OR queries. Stars and cliques can be identified with $O(n/\log n)$ additive queries or with at least $\Omega(n)$ OR queries. Our results summarized in Table \ref{tab:graphs} also confirm that parity  queries (and hence additive queries) are more powerful than OR queries in the quantum case.
Some other results include the following. Graphs with maximum degree $d$ can be learned with $O(dn)$ additive queries \cite{grebinski2000optimal}. 
This is also true for learning bipartite graphs with maximum degree $d$ non-adaptively \cite{bouvel05}.
Graphs with $m$ edges can be learned with $O(m(\log n)/(\log m))$ additive queries \cite{bshouty2011reconstructing, choi2010optimal}.
A general graph can be reconstructed with $\Theta(n^2/\log n)$ non-adaptive additive queries \cite{bouvel05}. 

Our results in the parity query model are closely related to a recent work by Lee, Santha and Zhang~\cite{lee20}. These authors showed that weighted graphs with maximum degree $d$ can be learned using $O(d \log^2 n)$ quantum ``cut queries'', and graphs with $m$ edges can be learned using $O(\sqrt{m} \log^{3/2} n)$ quantum cut queries. A cut query takes as input a subset $S$ of the vertices, and returns the number of edges of $G$ with exactly one endpoint in $S$. Lee, Santha and Zhang also gave efficient quantum algorithms in this model for determining the number of connected components of $G$, and for outputting a spanning forest of $G$. 
It was shown in~\cite{lee20} that cut queries reduce to additive queries; however, there is no efficient reduction in the other direction. In~\cite[Corollary 27]{lee20} stronger results than the cut-query results are given for additive queries: an $O(d \log(n/d))$ query algorithm for learning graphs with maximum degree $d$, and an $O(\sqrt{m \log n} + \log n)$ query algorithm for learning graphs with $m$ edges. These algorithms are based on very similar ideas to the ones we state here (Theorems \ref{thm:lowdegreegs} and \ref{thm:boundednumedges_parity}). Our algorithms as stated only require parity information (although the results of~\cite{lee20} could easily be rephrased in this way too); more importantly, the complexity of our results is somewhat better for graphs with very few edges, as a $\log n$ term is changed into a $\log m$ term. On the other hand, the algorithms of~\cite{lee20} are stated for the more general class of weighted graphs.

\textbf{Combinatorial group testing.}
Classically, it is known that the number of queries required to solve CGT is $\Theta(k \log(n/k))$  \cite{du2000combinatorial}.
In the quantum case,
Ambainis and Montanaro \cite{ambainis14b} first studied this problem  and proposed a quantum algorithm using $O(k)$ queries. They also showed a lower bound of $\Omega(\sqrt{k})$.
Later in~\cite{belovs15}, based on the adversary bound method, Belovs proved that there is a quantum algorithm that solves the CGT problem with $\Theta(\sqrt{k})$ queries.

%-------------------------------------------------------------------------------

\subsection{Preliminaries}
\label{sec:fs}

\textbf{Oracle models.} Let $f:\{0,1\}^n \rightarrow \{0,1\}$ be a boolean function with a quantum oracle to access it. That is, we are allowed to perform the map
\be \label{oracle}
\ket{x} \ket{y} \rightarrow \ket{x} \ket{y\oplus f(x)} 
\ee
for any $x\in\{0,1\}^n, y\in\{0,1\}$. Together with a simple phase flip unitary gate, we can also perform
\be \label{oracle: boolean function}
\ket{x}  \rightarrow (-1)^{ f(x)} \ket{x}.
\ee
For any $x\in\{0,1\}^n$, the Fourier coefficient of $f$ at $x$ is defined as
\be
\widehat{f}(x) = \frac{1}{2^n} \sum_{s \in \{0,1\}^n} (-1)^{f(s) + s\cdot x}.
\ee
We can equivalently associate each bit-string $x \in \{0,1\}^n$ with a subset $S \subseteq [n]$. The Fourier sampling primitive is based on the following sequence of operations: First apply Hadamard gates $H^{\otimes n}$ to $|0\rangle^{\otimes n}$; then apply the oracle (\ref{oracle: boolean function}); finally apply $H^{\otimes n}$ again.
The resulting state is
\be \label{resulting state of Fourier sampling}
\sum_{x\in \{0,1\}^n} \widehat{f}(x) \ket{x}.
\ee
Measuring in the computational basis returns $x$ with probability $\widehat{f}(x)^2$.

For an arbitrary graph $G = (V,E)$ on $n$ vertices, and an arbitrary subset $S \subseteq V$, define the oracles $\fgor$, $\fgpar$ by
\begin{itemize}
    \item $\fgor(S) = 0$ if $|E \cap (S \times S) | = 0$, and $\fgor(S) = 1$ otherwise;
    \item $\fgpar(S) = |\{E \cap (S \times S)\}| \text{ mod }2$.
\end{itemize}
We give quantum algorithms access to these oracles in the usual way, i.e., (\ref{oracle}) or (\ref{oracle: boolean function}).
%The graph state is defined as
%\be
%\ket{G} = \prod_{(i,j)\in E} CZ_{ij} \ket{+}^{\otimes n}.
%\ee

\textbf{Combinatorial group testing.} A subroutine that will be used extensively throughout this paper is Belovs' efficient quantum algorithm for combinatorial group testing (CGT)~\cite{belovs15}. In this problem, we are given oracle access to an $n$-bit string $A$ with Hamming weight at most $k$.
Usually, we assume that $k\ll n$. 
In one query, we can get the OR of an arbitrary subset of the bits of $A$. The goal is to determine $A$ using the minimal number of queries. Belovs showed that this can be achieved using $O(\sqrt{k})$ quantum queries.

\begin{thm}[Theorem 3.1 of~\cite{belovs15}]
\label{CGT algorithm}
The quantum query complexity of the combinatorial group testing problem is $\Theta(\sqrt{k})$. The quantum algorithm succeeds with certainty.
\end{thm}

For more details about this algorithm, refer to Section \ref{section:Combinatorial group testing}.

\textbf{Notation.} We sometimes use the notation $[X]$ for an expression which evaluates to 1 if $X$ is true, and 0 if $X$ is false.

\textbf{Outline of the paper.} 
In Section \ref{section:Learning an unknown graph with OR queries}, we consider the problem of learning graphs using OR queries.
In Section \ref{sec:graphstates}, we apply Bell sampling to learn unknown graph states.
In Section \ref{sec:parityqueries}, we investigate the problem of learning unknown graphs using parity queries.
In Section \ref{section:Combinatorial group testing}, we propose a time-efficient quantum algorithm for the combinatorial group testing problem.
In Section \ref{section:Majority and exact-half functions}, we apply Fourier sampling to produce time-efficient quantum algorithms for learning majority and exact-half functions.

%-------------------------------------------------------------------------------

\section{Learning an unknown graph with OR queries}
\label{section:Learning an unknown graph with OR queries}

Let $G$ be a graph with $m$ edges and $n$ vertices. Our goal is to identify all the edges in $G$ using OR queries. We follow the same general strategy as Angluin and Chen~\cite{angluin08} to achieve this by starting with special cases and progressively generalising.  In particular, Lemmas \ref{lem:learn graph 1} and \ref{lem:learn graph 2} are direct quantum speedups of corresponding results (Lemmas 3.3 and 3.4) in~\cite{angluin08}.
The basic idea of the quantum learning algorithm is as follows: We first decompose the set of vertices into a disjoint union of several subsets. Each subset contains no edges. Then we learn the edges between these subsets. A sub-routine of this learning procedure is the quantum algorithm for solving solving combinatorial group testing (CGT), i.e., Theorem \ref{CGT algorithm}. It is the main ingredient to obtain quantum speedups.

Suppose $A,B$ are two known, nonempty, independent (i.e., contain no edges) subsets of the set of vertices. The following lemma helps us efficiently identify the non-isolated vertices (those which have at least one edge incident to them). 

\begin{lem}
\label{lem:learn non-isolated vertices}
Assume that $A$ and $B$ are two known, disjoint,
non-empty independent sets of vertices in $G$.
Suppose there are $n_A, n_B$ non-isolated vertices in
$A$ and $B$ respectively.
Then there is a quantum algorithm that identifies these
non-isolated vertices with $O(\sqrt{n_A}+\sqrt{n_B})$
OR queries. The algorithm succeeds with certainty.
\end{lem}

\begin{proof}
For each subset $S \subseteq A$, we consider queries of the form $S \cup B$.
The result is 1 if and only if there is a non-isolated 
vertex in $S$. We can view $A$ as a bit string such that
the $i$-th element is 1 if the $i$-th vertex is non-isolated,
and 0 otherwise. Using the quantum algorithm for CGT (see Theorem \ref{CGT algorithm}), we can learn this
bit-string with $O(\sqrt{n_A})$ queries.
Similarly, we can learn the non-isolated vertices in $B$
with $O(\sqrt{n_B})$ queries.
\end{proof}

Note that if there are $m_{AB}$ edges between $A,B$, then $n_A, n_B \leq \min(m_{AB}, n) \leq \min(m,n)$.
Next, we show how to learn the edges between $A$ and $B$. 
Lemma \ref{lem:learn graph 1} below focuses on a general case, and Lemma \ref{lem:learn special graph 1} considers the case of bounded-degree graphs.

\begin{lem}
\label{lem:learn graph 1}
Make the same assumptions as Lemma \ref{lem:learn non-isolated vertices}.
Suppose there are $m_{AB}$ edges between $A$ and $B$.
Then there is a quantum algorithm that identifies these
edges with $O(m_{AB})$ OR queries. The algorithm succeeds with certainty.
\end{lem}

\begin{proof}
By Lemma \ref{lem:learn non-isolated vertices}, we assume that there are no isolated vertices in $A,B$. It costs 1 query to check if $m_{AB}=0$ or not. In the following, we shall assume that $m_{AB}>0$.
We view each vertex as a variable.
Then the learning problem is equivalent to learn 
the Boolean function
$f=x_1f_1\vee \cdots \vee x_{n_A} f_{n_A}$, 
where $f_1,\ldots,f_{n_A}$ are OR functions
of variables $y_1,\ldots,y_{n_B} \in B$,
and where  $x_1,\ldots,x_{n_A} \in A$. 
To learn $f$, we first set all variables in $B$ to 1,
then $f$ becomes $x_1\vee\cdots\vee x_{n_A}$.
By the CGT algorithm, we can learn  $x_1,\ldots,x_{n_A}$
with $O(\sqrt{n_A})$ queries.
Next, for each $i\in\{1,\ldots,n_A\}$,
we set $x_i=1, x_j = 0$ $(j\neq i)$, then
we are left with $f_i$. 
Using the CGT algorithm again,
we can learn $f_i$ with $O(\sqrt{a_i})$ queries,
where $a_i$ is the size of $f_i$, i.e., the number of
relevant variables in $f_i$.
Thus the total number of queries is
\bes
O(\sqrt{n_A} + \sqrt{a_1} + \cdots 
+ \sqrt{a_{n_A}} ).
\ees
Since $a_1+\cdots+a_{n_A} = m_{AB}$ and $n_A \le m_{AB}$,
the number of queries is bounded by
$
O(\sqrt{m_{AB}} + m_{AB}) = O(m_{AB}),
$
which is tight when $a_1=\cdots=a_{n_A}=1$ and ${n_A}=m_{AB}$.
\end{proof}

When the graph is bounded-degree, the above lemma can be improved. 

\begin{lem}
\label{lem:learn special graph 1}
Make the same assumptions as Lemma \ref{lem:learn graph 1}, and additionally suppose that $G$ has maximum degree $d$. Then there is a quantum algorithm that identifies the edges in $G$ using $O(d^2 \sqrt{m_{AB}}\log m_{AB})$ OR queries. The algorithm succeeds with certainty.
\end{lem}

\begin{proof}
For simplicity, let each vertex in $A$ and $B$ have an index in the set $\{1,\dots,|A|\}$, $\{1,\dots,|B|\}$ respectively.
By Lemma \ref{lem:learn non-isolated vertices}, we can assume that
$|A| = n_A, |B| = n_B$. That is, there are no isolated vertices in $A,B$.

First, to gain intuition, we consider the special case of matchings ($d=1$). 
In this case, $n_A=n_B=m_{AB}\leq n$.
For each $a \in A$, we use $n_a$ to denote the index of the neighbour of $a$ in $B$, if such a neighbour exists, and otherwise set $n_a = 0$.
For any $T \subseteq B$, let $B^T \in \{0,1\}^{|A|}$ denote the bit-string whose $i$'th element equals 1 if $n_i \in T$, and 0 otherwise.
Fixing the same $T$ and varying over subsets $S\subseteq A$ and queries of the form $S\cup T$, 
we can think of this oracle query as returning 1 if there exists $i \in S$ 
such that $n_i \in T$ (equivalently, $B^T_i = 1$), and 0 otherwise. 
This is the same oracle used in CGT, 
so this means that $B^T$ can be learned completely using $O(\sqrt{|B^T|})$ quantum queries for any fixed $T$.
Here $|B^T|$ is the Hamming weight of the bit-string $B^T$.

We then repeat this algorithm for different choices of $T$. 
In particular, we can think of each $n_i \in \{1,\dots,|B|\}$ as an element of $\{0,1\}^{\lceil\log(|B| + 1)\rceil}$, 
and consider the sequence $T_j = \{i: i_j=1\}$, $j=1,\dots,\lceil\log(|B|+1) \rceil$. 
Then $k := \lceil\log(|B|+1) \rceil = O( \log m_{AB})$ repetitions are enough to learn 
all the bits of $n_i$ for all $i \in A$, and hence to learn the graph completely. 
The overall complexity is 
\[
O\left(\sqrt{|B^{T_1}|} + \cdots + \sqrt{|B^{T_k}|} \right).
\]
As $|B^{T_i}| \le m_{AB}$ for all $i$,
%For matchings, {\color{red} $|B^{T_1}| + \cdots + |B^{T_k}| = m_{AB}$} \am{I'm not sure I now see why this claim holds. The $T$ sets intersect, so it could be the case that an edge is included in more than one of them. So I think the upper bound could be as big as $\sqrt{m_{AB}} k$...},
the complexity is bounded by $O(\sqrt{m_{AB}} \log m_{AB})$.
Note that there is no need to repeat the CGT algorithm to
reduce its error probability, as it is already exact.

Next, we consider bounded-degree graphs.
We can generalise the above idea to learning bipartite graphs where every vertex in $A$ has degree at
most $d$. For each $a\in A$, we now define $n_a$
as the set of the indices of the neighbours of $a$ in $B$. 
For any $T \subseteq B$, define $B^T \in \{0,1\}^{|A|}$ as the bit-string such that the $i$'th element equals 1 if $n_i \cap T \neq \emptyset$, and 0 otherwise.
Then, for any choice of $T$, an oracle query of the form $S \cup T$, $S \subseteq A$, returns whether any vertex in $S$ has any neighbours in $T$. This implies that $B^T$ can be learned with $O(\sqrt{|B^T|})$
queries using the quantum algorithm for CGT~\cite{belovs15}.

There are randomised constructions of families of subsets $T$ of size {$k = O(d^2 (\log n_B))$} 
that 
allow the $d$ nonzero entries to be determined deterministically, for any pattern of nonzero 
entries (these are ``nonadaptive'' combinatorial group testing schemes \cite{aldridge19,atia2012boolean,porat2008explicit}).  
%\cite{chan2011non}
Since $|B^{T_i}|\leq m_{AB}$ for all $i$,
the overall complexity is $O(d^2 \sqrt{m_{AB}}(\log n_B))$. 
\end{proof}

There are also nonadaptive combinatorial group testing strategies that are designed to have a low worst-case probability of error~\cite{aldridge19,chan2011non}, and have only a linear dependence on $d$. However, it is not clear that these schemes can be used in our setting, as the failure probability would be of the form $n_B^{-\delta}$, for some $\delta > 0$, and $n_B$ might be much less than $n$.

In the following, we consider a more general case when $A,B$ are not independent.

\begin{lem}
\label{lem:learn graph 2}
Assume that $A$ and $B$ are two disjoint,
non-empty sets of vertices in $G$ with $m_A, m_B$ known edges respectively.
Suppose there are $m_{AB}$ edges between $A$ and $B$. Then there is a quantum algorithm that identifies these
edges using $O(m_{AB}+m_A+m_B)$ OR queries.
In particular, if $G$ has maximal degree $d$ and is colorable with $O(1)$ colors, then the algorithm uses
$O(d^2\sqrt{ m_{AB} }\log  m_{AB})$ queries. 
%$O(d\sqrt{\min(m_{AB},n)\log \min(m_{AB},n)})$ queries. 
%If $G$ is a Hamiltonian cycle, then the algorithm uses $O (\sqrt{m_{AB} \log m_{AB}})$ queries.
\end{lem}

\begin{proof}
The idea behind the quantum algorithm is as follows:
We first color the two graphs induced by $A, B$ such that
each color class is an independent set in $G$. 
Then we use Lemmas  \ref{lem:learn graph 1} and \ref{lem:learn special graph 1}
to identify the edges between 
color classes in $A$ and color classes in $B$.

It is well-known that a graph with $t$ edges can be $\lfloor \sqrt{2t} + 1\rfloor$-colored. The coloring can be constructed in polynomial time.
Now let $q_1 = \lfloor \sqrt{2m_A} + 1\rfloor, 
q_2 = \lfloor \sqrt{2m_B} + 1\rfloor$ be the number of colors used for $A$ and $B$, respectively.
Assume that there are $m_{ij}$ edges between the $i$-th color 
class of $A$ and the $j$-th color class of $B$.
Then by Lemma \ref{lem:learn graph 1},
the number of queries used to identify the
edges between $A$ and $B$ is bounded by
$
\sum_{i=1}^{q_1} \sum_{j=1}^{q_2}
O (m_{ij} + 1)
=O (m_{AB} +  q_1 q_2 )
=O(m_{AB}+m_A+m_B).
$

If $G$ is $O(1)$-colorable, then $q_1,q_2=O(1)$. By Lemma \ref{lem:learn special graph 1} and the same argument as above, if $G$ additionally has maximal degree $d$, all edges can be identified with
$O(d^2\sqrt{ m_{AB}  } \log m_{AB})$
%$O(d\sqrt{\min(m_{AB},n)\log \min(m_{AB},n)} )$ 
queries.
\end{proof}

The next lemma generalizes the above lemma  to learn the edges of multiple disjoint subsets. Note that if there are $k$ subsets, then there are $O(k^2)$ pairs. So naively we need to make at least $O(k^2)$ queries. However, this can improved to be linear in $k$ by using Lemma \ref{lem:learn graph 2} in a binary decomposition approach.

\begin{lem}
\label{lem:learn graph 3}
Assume that $S_0,\ldots,S_{k-1}$ are disjoint non-empty sets of vertices in $G$, and
each has $s_i$ known edges. 
Suppose there are $s_{i,j}$ edges between $S_i$ and $S_j$, then there is a quantum algorithm that
identifies all the edges using $O(k + T \log k)$ OR queries,
where $T = \sum_i s_i + \sum_{i,j} s_{i,j}$.
If $G$ has maximal degree $d$ and is $O(1)$-colorable, then the number of queries can be reduced to 
$O(k + d^2 \sqrt{kT} \log T )$.
\end{lem}

\begin{proof}
For simplicity, we assume that $k=2^l$ for integer $l$.
Set $K = \sum_{i=0}^{k-1} s_i$.
The idea of the algorithm is to recursively use Lemma \ref{lem:learn graph 2} in a binary form.
In step 1, for each pair $(S_{2i}, S_{2i+1})$,
we use Lemma \ref{lem:learn graph 2} to find
the edges between them. There are $2^{l-1}$ pairs in total. So this step uses
\[
O\Bigg(
\sum_{i=0}^{2^{l-1} - 1} (s_{2i,2i+1} + s_{2i} + s_{2i+1} + 1)
\Bigg)
=O\Bigg(2^{l-1} + K + 
\sum_{i=0}^{2^{l-1} - 1} s_{2i,2i+1}  
\Bigg)
\]
queries in total.
After step 1, we know the edges of each adjacent pair
$(S_{2i}, S_{2i+1})$. So we can combine them and 
obtain a new set, denoted as $S_i':=S_{2i} \cup S_{2i+1}$ for $i=0,1,\ldots,2^{l-1}-1$.
It has $s_i':=s_{2i,2i+1} + s_{2i} + s_{2i+1}$ edges.
The number of edges between $S_i'$ and $S_j'$ is
$s_{i,j}':=s_{2i,2j}+s_{2i,2j+1}+s_{2i+1,2j}+s_{2i+1,2j+1}$. Now, similarly to step 1, we can learn
the edges between $(S_{2i}', S_{2i+1}')$.
This step uses
\beas
&& O\Bigg(2^{l-2} + K + 
\sum_{i=0}^{2^{l-1} - 1} s_{2i,2i+1}  
+
\sum_{i=0}^{2^{l-2} - 1} s_{4i,4i+2}+s_{4i,4i+3}+s_{4i+1,4i+2}+s_{4i+1,4i+3}
\Bigg)
\eeas
queries in total.
Continuing the above procedure, we can learn all the edges. The above procedure terminates after $l = O(\log k)$ steps.
It is not hard to show that
the total number of queries is bounded by
\beas
O(K l + 2^l + T_{1}  + T_{2} 
+\cdots + T_{l-1}),
\eeas
where $T_i$ is the total number of edges between two adjacent pairs in step $i$.
Since $T_i \leq T- K$,
the number of queries is bounded by
$O(T(\log k) + k)$.

When $G$ has maximal degree $d$ and is $O(1)$-colorable, by Lemma
\ref{lem:learn graph 2}, the number of queries used in 
step $i$ is 
% \[
% O\Bigg(
% 2^{l-i} + 
% d \sum_{j=0}^{2^{l-i} - 1} \sqrt{ \min(s_{2j,2j+1}^{(i)}, n) \log \min(s_{2j,2j+1}^{(i)},n) }  
% \Bigg),
% \]
\[
O\Bigg(
2^{l-i} + 
d^2 \sum_{j=0}^{2^{l-i} - 1} \sqrt{  s_{2j,2j+1}^{(i)}   }  \log  s_{2j,2j+1}^{(i)} 
\Bigg),
\]
where $s_{2j,2j+1}^{(i)}$ is the number of edges of the $j$-th adjacent pair in step $i$.
It is easy to check that 
$\sum_{i=1}^{l-1} \sum_{j=0}^{2^{l-i} - 1} s_{2j,2j+1}^{(i)} = T - K$,
thus the total number of queries used in the algorithm is bounded by $O(k + d^2 \sqrt{kT} \log T )$, where we bound $s_{2j,2j+1}^{(i)} \le T$ and use Cauchy-Schwarz inequality.
\end{proof}

We can now use these ingredients to obtain algorithms for learning general graphs using OR queries. By the above lemma, what remains is to decompose the set of vertices into a disjoint union of a small number of subsets. We shall use the following trick described in \cite{angluin2006learning}.

Given a probability $p$, a $p$-random set $S$ is obtained by including each vertex independently with probability $p$. Then the probability that a $p$-random set includes no edge of $G$ is at least $q=1-mp^2$.
Choosing $p=1/10\sqrt{m}$, then the probability 
is at least $q=0.99$.
The size of $S$ is close to
$pn$ with high probability. 

Let $V$ denote the set of vertices.
First we identify a random set $S_1$ that includes no edge of $G$ by following the above procedure. 
After we have $S_1$, then in $V-S_1$,
we can find another random set $S_2$ that includes no edge of $G$.
We continue this process for $k$ steps, where $k$ is determined later.
Assume now that we have $k$ random sets $S_1,\ldots,S_k$. Each has no
edge of $G$. This uses $O(k)$ queries in total.
After $k$ steps, the number of remaining vertices is about $(1-p)^k n \approx e^{-pk}n$. This means that the above procedure terminates after $k = O(p^{-1} \log n) = O(\sqrt{m} \log n)$ steps with high probability.

\begin{thm}
\label{thm:graphorqueries}
Let $G$ be a graph with $m$ edges and $n$ vertices. Then
there is a quantum algorithm that learns
the graph by making
\be
O\left(m\log(\sqrt{m} \log n)+ \sqrt{m} \log n \right)
\ee
OR queries with probability at least 0.99. If $G$ has maximal degree $d$ and is $O(1)$-colorable,  the query complexity is
\be
O\left(d^2 m^{3/4} \sqrt{\log n} (\log m) + \sqrt{m} \log n\right)
\ee
with probability at least 0.99.
% \be
% O\left(d m^{3/4} \sqrt{(\log n)(\log m)} \right).
% \ee
\end{thm}

\begin{proof}
The idea of our algorithm is as follows: we first decompose the vertices of the graph into $k=O(\sqrt{m} \log n)$ independent subsets by the above arguments. Then we learn the edges among all the pairs using Lemma \ref{lem:learn graph 3}.

The first step uses $O(\sqrt{m} \log n)$ OR queries.
By Lemma \ref{lem:learn graph 3}, all the edges can be identified with 
$$
O(k + m \log k) = O(\sqrt{m} (\log n) + m\log (\sqrt{m} \log n) )
$$
queries. If the graph has maximal degree $d$ and is $O(1)$-colorable, then the number of queries is 
% \[
%  O\left(d\sqrt{m} (\log n) 
% \sqrt{\min(\frac{\sqrt{m}}{\log n},n)\log \min(\frac{\sqrt{m}}{\log n},n)} + \sqrt{m} \log n \right)  
% = O\left(d m^{3/4} \sqrt{(\log n)(\log m)} \right).
% \]
\[
O\left(d^2 
\sqrt{m^{3/2} (\log n) } (\log m) + \sqrt{m} \log n \right)  
= O\left(d^2 m^{3/4} \sqrt{\log n} (\log m) + \sqrt{m} \log n\right).
\qedhere
\]
\end{proof}

The quantum query complexity achieved by the first part of Theorem \ref{thm:graphorqueries} is an improvement over the $\Omega(m \log (n^2/m))$ classical lower bound if $m$ is very small with respect to $n$; for example, if $m=\Theta(\log n)$, the complexity is $O(\log^{1.5} n )$, as compared with $\Omega(\log^2 n)$ classically. However, if $m = \Omega(n^\epsilon)$ for some fixed $\epsilon>0$, the complexity is worse than the classical lower bound.

If $G$ is promised to be a Hamiltonian cycle or a matching (for example), then $d=O(1)$, and by the second part of Theorem \ref{thm:graphorqueries} the number of OR queries used to learn $G$ is bounded by $O(m^{3/4} \sqrt{\log n}(\log m) + \sqrt{m}\log n)$, which is an improvement over the $\Omega(m\log (n/m))$ classical complexity for large $m$.

\subsection{Learning specific graphs using OR queries}

Next we give quantum algorithms for learning certain specific graph families using OR queries.

\begin{prop}
There is a quantum algorithm which makes $O(\sqrt{k})$ OR queries and identifies an arbitrary clique on $k$ vertices.
\end{prop}

\begin{proof}
The idea is as follows: First, we find a vertex $v$ in the clique, then use the quantum algorithm for CGT~\cite{belovs15} to learn all the other vertices using $O(\sqrt{k})$ queries, by querying with subsets of the vertices that include $v$. Such a query returns 1 if and only if the subset includes another vertex of the clique.

As for the first step,
the vertex $v$ can be found with high probability using $O(1)$ queries, using a similar idea to the quantum algorithm of~\cite{ambainis14b} for CGT.
We produce a subset $S$ of vertices by including each vertex with probability $1/k$. Then with probability $\binom{k}{2} k^{-2} (1-1/k)^{k-2}
\approx 1/2e$, this leads to exactly 2 vertices $i,j$ in the clique being included in the subset. 
This subset corresponds to a boolean function $f(x) = x_ix_j$ for unknown $i,j$.
To learn $i,j$, we use the Fourier sampling method.
Let $b_k$ be the bit-string of length $n$ whose $k$-th bit equals 1, and all other bits equal 0.
It is easy to verify that the Fourier coefficients of $f$ at $ b_i, b_j, b_i+b_j$ are all equal to $1/2$.
Thus with probability at least $3/4$, we can identify $x_i$ or $x_j$.
\end{proof}

% Suppose for bipartite graph,
% the result of Lemma 11 is
% $O(m^{\alpha})$ queries. Then
% the complexity in Lemma 12 becomes
% $O(m^{\alpha})$,
% and the complexity in Lemma 12 becomes
% $O(k+m^\alpha k^{1-\alpha} \log m)$.
% Then by Theorem 14, $k=\sqrt{m}\log n$,
% so the complexity to learn bipartite graph
% is $\widetilde{O}(m^{(1+\alpha)/2})$.
% This is at least $\widetilde{O}(m^{3/4})$.

\begin{prop}
\label{prop:Learning star graph}
There is a quantum algorithm which makes $O(\sqrt{m})$ OR queries and identifies an arbitrary star graph with $m$ edges.
\end{prop}

\begin{proof}
This is equivalent to learning the Boolean function 
$f(x) = x_i \wedge (\vee_{j\in A} x_j)$, for some unknown $i$, $A$,
where $A$ is a subset of $[n]$ of size $m$ 
and $i\notin A$.
To learn it, we use the Fourier sampling to identify the center $x_i$ first, then use the CGT algorithm to learn the edges.

The Fourier sampling method returns a state of the form
$
\sum_{y\in \{0,1\}^n} \widehat{f}(y) \ket{y}.
$
Consider the Fourier coefficient at $y_i=1, y_j = 0$ $(j\neq i)$.
It equals
\beas
\frac{1}{2^n}\sum_{x\in \{0,1\}^n} (-1)^{x_i \wedge (\vee_{j\in A} x_j) + x_i} 
&=&
\frac{1}{2^n} \left(\sum_{x\in \{0,1\}^n:x_i = 0} 1 - \sum_{x\in \{0,1\}^n: x_i = 1} (-1)^{\vee_{j\in A} x_j} \right)\\
&=& \frac{1}{2^n} (2^{n-1} - (2 - 2^{m-1})2^{n-m} ) \\
&=& 1 - \frac{1}{2^{m-1}}.
\eeas
This means that Fourier sampling can detect the center with $O(1)$ queries with high probability.
After we obtain the center, it suffices to
focus on the function obtained by setting $x_i = 1$.
Using the quantum algorithm for CGT, 
we can learn this function with $O(\sqrt{m})$ queries.
\end{proof}

The above two results are tight because of the optimality of CGT. More precisely, CGT corresponds to the special case of learning a clique when one vertex is given, or learning a star when the center is given.

%-------------------------------------------------------------------------------

\subsection{Lower bound}

Finally, we show a quantum lower bound for learning graphs with OR queries, which shows that the quantum algorithm given in Theorem \ref{thm:graphorqueries} for learning graphs with $m$ edges is optimal up to a logarithmic factor.

\begin{thm}
\label{thm:orqlower}
Let $G$ be an arbitrary graph of $n$ vertices. Then any quantum algorithm that learns $G$ with success probability $>1/2$ using OR queries must make $\Omega(n^2)$ queries.
\end{thm}

\begin{proof}
Consider the family of graphs on $2n$ vertices defined as follows. We first start with two disjoint cliques $A$, $B$ on $n$ vertices. We then put edges between $A$ and $B$ in an arbitrary pattern. This corresponds to an adjacency matrix of the form
\[ \begin{pmatrix} J-I & M\\ M^T & J-I \end{pmatrix} \]
where $J$ is the all-1's matrix, and $M$ is an arbitrary $n \times n$ matrix. Now observe that any query that contains more than one vertex in $A$, or more than one vertex in $B$, will always return 1. Any query that contains only one vertex in total will always return 0. So we can restrict to considering queries that include exactly one vertex of $A$ and exactly one vertex of $B$. Such a query just returns one of the entries of $M$. Learning $M$ with success probability $>1/2$ using this oracle requires $\Omega(n^2)$ quantum queries~\cite{beals01}.
\end{proof}

As a corollary, we get the lower bound that any quantum algorithm that learns an arbitrary graph with $m$ edges must make $\Omega(m)$ quantum queries. Also, by the known lower bound on the quantum query complexity of the parity function~\cite{beals01}, if $m$ is unknown, then any quantum algorithm that determines $m$ exactly must make $\Omega(m)$ queries when $m=\Omega(n^2)$.

%-------------------------------------------------------------------------------

\section{Learning an unknown graph state}
\label{sec:graphstates}

The graph state $\ket{G}$ on $n$ qubits corresponding to a graph $G=(V,E)$ with $n$ vertices can be defined explicitly as
\be \label{eq:graphstate} 
\ket{G} = \frac{1}{\sqrt{2^n}} \sum_{x \in \{0,1\}^n} (-1)^{\sum_{(i,j) \in E} x_i x_j} \ket{x},
\ee
The state
$\ket{G}$ can also be defined as the state produced by acting on the uniform superposition $\ket{+}^{\otimes n}$ with a controlled-$Z$ gate across each pair of qubits corresponding to an edge in $G$, or as the unique state stabilized by the set of Pauli operators $\{ X_v \prod_{w \in N(v)} Z_w : v \in V \}$, where $N(v)$ denotes the set of vertices neighbouring $v$~\cite{hein06}.

The representation (\ref{eq:graphstate}) makes it clear that graph states have a close connection to the parity query model, as $\ket{G}$ is the state produced by evaluating $\fgpar(S)$ on all subsets $S$ in uniform superposition. Therefore, lower bounds on the complexity of identifying graphs using parity queries imply lower bounds on the number of copies of $\ket{G}$ required to identify $G$, and upper bounds on the number of copies of $\ket{G}$ required to identify $G$ imply upper bounds on the complexity of identifying $G$ using parity queries.

First we show how to partially go in the other direction, by making parity queries of a certain form, given copies of $\ket{G}$. We use a procedure called Bell sampling, which was used for learning arbitrary stabilizer states in~\cite{montanaro17a}. Given two copies of a state $\ket{\psi}$ of $n$ qubits, Bell sampling corresponds to measuring each corresponding pair of qubits in the Bell basis. Outcomes of Bell sampling can be identified with strings $s \in \{I,X,Y,Z\}^n$ of Pauli matrices, and are observed with the following probabilities:

\begin{lem}[{Lemma 2 of \cite{montanaro17a}}]
    \label{lem:bsamp}
    Let $\ket{\psi}$ be a state of $n$ qubits. Bell sampling applied to $\ket{\psi}^{\otimes 2}$ returns outcome $s$ with probability
    \[ \frac{|\braket{\psi|\sigma_s|\psi^*}|^2}{2^n}, \]
    where $\ket{\psi^*}$ is the complex conjugate of $\ket{\psi}$ with respect to the computational basis, and $\sigma_s = s_1 \otimes s_2 \otimes \dots \otimes s_n$.
\end{lem}

If $\ket{G}$ is a graph state, then $\ket{G} = \ket{G^*}$, and $|\braket{G|\sigma_s|G}|^2=1$ if and only if $\sigma_s$ is a stabilizer of $\ket{G}$; otherwise, $|\braket{G|\sigma_s|G}|^2=0$. Therefore, Bell sampling returns a uniformly random stabilizer of $\ket{G}$. Such a stabilizer can be produced by taking the product of a random subset $S$ of the rows of the stabilizer matrix for $G$ (where each row is included with independent probability $1/2$). We obtain the following overall operator:
\[  \prod_{v \in S} X_v \prod_{u \in N(v)} Z_u = \pm \prod_{u \in [n]} X_u^{[u \in S]} Z_u^{|N(u) \cap S|} \]
where we collect $X$ and $Z$ terms together for each vertex $u \in [n]$. Hence, when we receive a sample of a uniformly random stabilizer of $\ket{G}$, we obtain a random subset $S \subseteq [n]$, and for each $u \in [n]$, we learn the number of edges between $u$ and $S$, mod 2. We learn the identity of $S$ from which qubits have an $X$ term associated with them.

This allows us to try to find efficient algorithms based only on this (now classical) subroutine of learning subsets and parities. Indeed, learning a graph state using Bell sampling is equivalent to learning a graph using parity queries, as studied in Section \ref{sec:parityqueries} below -- except with the restriction that these queries are only to uniformly random subsets of the vertices. We first give a general algorithm for learning a graph known to be picked from any finite set.

\begin{thm}
\label{thm:learning}
    Let $S$ be a family of graphs. Then, for any $G \in S$, $G$ can be identified by applying Bell sampling to $O(\log |S|)$ copies of $\ket{G}$. The algorithm succeeds with probability at least 0.99.
\end{thm}

\begin{proof}
    Let $A$ be the adjacency matrix of $G$. Each Bell sample returns the inner product of a random vector $\mathbf{s} \in \F_2^n$ with each column (or row) of $A$. If we take $k$ samples, we can write these $k$ row vectors as an $n \times k$ matrix $B$. Then the result of the Bell sampling procedure is the matrix $AB$.
    
    To be able to uniquely identify $G$, we want $AB \neq A'B$ for all $A$, $A'$ corresponding to graphs in $S$, or in other words $(A+A')B \neq 0^{n\times k}$. As each entry of $B$ is uniformly random, for any  $n\times n$ matrix $C$ with rank $r$, $\Pr_B[CB = 0^{n\times k}] = 2^{-kr}$. (This holds because for each linearly independent row $\mathbf{c}$ of $C$, $\Pr_B[\mathbf{c}B = 0^k] = 2^{-k}$, and these events are independent.) In particular, for any nonzero matrix $C$, $\Pr_B[
   CB = 0^{n\times k}] \le  2^{-k}$. The number of matrices $C$ of the form $C = A + A'$ is at most $|S|^2$. Taking a union bound over all such matrices, we have
    \[ \Pr_B[\exists C = A + A', CB = 0^{n\times k}] \le \frac{|S|^2}{2^k}. \]
    So it is sufficient to take $k = O(\log |S|)$ to achieve failure probability $0.01$, as claimed.
\end{proof}

As a corollary of Theorem \ref{thm:learning}, if $G$ is a graph with at most $m$ edges, it can be identified with $O(m \log(n^2/m))$ copies of $\ket{G}$.

It is natural to wonder whether the dependence on $|S|$ in Theorem \ref{thm:learning} could be improved, because if $S$ is the set of all graphs, the complexity of Theorem \ref{thm:learning} does not match that of the best algorithms for learning an arbitrary graph state, which use $O(n)$ copies of $\ket{G}$~\cite{aaronson08,zhao16,montanaro17a}. An information-theoretic lower bound comes from the fact that $\ket{G}$ is a state of $n$ qubits, so by Holevo's theorem, $\Omega((\log |S|)/n)$ copies are required to identify a state from $S$. In addition, this bound cannot always be reached; if $S$ is the set of all graphs on $r$ vertices, for some $r < n$, the number of copies required to identify a graph from this set is $\Theta(r)$ by the same information-theoretic argument, which can be much larger than $O(r^2/n)$ for some choices of $r$. This suggests that the best dependence on $|S|$ that could be achieved is $O(\sqrt{\log|S|})$.

However, better complexities can be achieved for graphs with more structure.
If the graph is promised to be a star, then the Fourier sampling method can be applied to learn it with $O(1)$ copies of $\ket{G}$. More precisely,
suppose the edges of the star graph are $(i, j), j \in A$. Here $i$ is the center and we assume $|A|\geq 1$. Then
\[
\ket{G} = \frac{1}{\sqrt{2^n}} \sum_{x\in\{0,1\}^n} (-1)^{x_i \sum_{j\in A} x_j} \ket{x}.
\]
By Fourier sampling, if we apply Hadamard gates to $\ket{G}$, we obtain the state
\be
\frac{1}{\sqrt{2}} \ket{0,\ldots,0} \ket{+} \ket{0,\ldots,0} + \frac{1}{\sqrt{2}} \ket{[1\in A],\ldots,[i-1\in A]} \ket{-} \ket{[i+1\in A],\ldots,[n\in A]}.
\ee
%where $[j\in A] = 1$ if $j \in A$ and 0 otherwise.
The $\ket{\pm }$ is in the $i$-th qubit. Performing measurements in the computational basis, if we obtain $\ket{0,\ldots,0} \ket{1} \ket{0,\ldots,0}$, then we know the center; if we obtain a state
with more than two 1's, then we know all vertices in $A$. The probability is 1/4 for each case, so we can learn a unknown star using $O(1)$ copies of $\ket{G}$.

We can also apply Bell sampling to learn cliques with $O(1)$ copies. Each Bell sample gives us the inner product of each row of the adjacency matrix with a random vector, and each nonzero row has probability $1/2$ for this inner product to be nonzero. As $G$ is a clique, all its nonzero rows are the same. Thus, after $O(1)$ samples, with high probability we learn all the nonzero rows at once.

In summary, we have
\begin{thm}
There is a quantum algorithm that identifies $G$ by using $O(1)$ copies of $\ket{G}$ if $G$ is a star or a clique.
\end{thm}

Next we consider the case of bounded-degree graphs.

\begin{thm}
\label{thm:lowdegreegs}
For an arbitrary graph $G$, there is a quantum algorithm which uses $O(d \log(m/d))$ copies of $\ket{G}$, and for each vertex $v$ that has degree at most $d$, outputs all the neighbours of $v$ and that $v$ has degree at most $d$. For each vertex $w$ that has degree larger than $d$, the algorithm outputs ``degree larger than $d$''. The algorithm succeeds with probability at least 0.99.
\end{thm}

%The proof is essentially the same as Lemma \ref{lem:lowdegree}. \am{Best way of presenting - e.g. just refer to other result?}

\begin{proof}
We assume that $d \le n/4$ throughout, as otherwise an algorithm for learning an arbitrary graph using $O(n)$ copies can be used~\cite{aaronson08,zhao16,montanaro17a}. We produce $k$ Bell samples, corresponding to vectors $A \w_1,\dots,A \w_k$ for uniformly random vectors $\w_1,\dots,\w_k \in \{0,1\}^n$.
For any pair $\x \neq \y \in \{0,1\}^n$, the probability that $\x\cdot \w_i = \y \cdot \w_i$ for all $i$ is equal to the probability that $(\x+\y)\cdot \w_i = 0$ for all $i$, which equals $2^{-k}$. By a union bound, for any $\x \in \{0,1\}^n$, the probability that there exists $\y \in \{0,1\}^n$ such that $\y \neq \x$, $|\y| \le d$ and $\x\cdot \w_i = \y \cdot \w_i$ for all $i$ is bounded by $\sum_{l=0}^d \binom{n}{l} 2^{-k} = O(2^{d\log(n/d)-k})$.

We then apply this bound to all $n$ rows of $A$ via a union bound, to obtain that the probability that, for any row $\x$ of $A$, there exists $\y \in \{0,1\}^n$ with $|\y| \le d$, $\x\cdot \w_i = \y \cdot \w_i$ for all $i$ and $\y \neq \x$ is $O(n 2^{d \log(n/d)-k})$. Taking $k = O(d \log(n/d))$ is sufficient to bound this probability by an arbitrarily small constant.
Assuming that this failure event does not occur, the algorithm determines all rows of $A$ with Hamming weight bounded by $d$, and identifies all rows that are inconsistent with having Hamming weight bounded by $d$.

We finally show how to replace $n$ with $m$ in the algorithm's complexity. This is achieved by first identifying the subset $W$ of non-isolated vertices, and then running the algorithm above on the vertices in this subset. We can restrict the graph to this subgraph $H$ by measuring the qubits corresponding to the other vertices in the computational basis. The resulting state is of the form $\ket{H'} = \prod_{i \in T} Z_i \ket{H}$, for some subset $T \subseteq W$. By Lemma \ref{lem:bsamp}, Bell sampling behaves in the same way on $\ket{H'}$ as on $\ket{H}$.
To find the subset $W$, Bell sampling is applied $l$ times for some $l$, to produce an $n \times l$ matrix $C = AB$ for a uniformly random matrix $B$. The set of vertices corresponding to rows of $C$ which have at least one nonzero entry is kept, to produce a set $W'$. Any zero row of $A$ will always produce a corresponding zero row of $C$, so will not be included in $W'$. On the other hand, the probability that any nonzero row of $A$ produces the corresponding zero row of $C$ is $2^{-l}$. As there are at most $2m$ nonzero rows, corresponding to vertices in $W$, the probability that any vertex in $W$ is not included in $W'$ is $O(m 2^{-l})$ by a union bound. So it is sufficient to take $l=O(\log m)$ to learn which rows are nonzero with probability $0.99$.
\end{proof}

We can also learn the family of graphs that are subgraphs of a fixed graph $G'$ of bounded degree $d$. This is relevant to the setting where we have attempted to produce $\ket{G'}$ using a quantum circuit which may have failed to produce certain edges, and we would like to determine which graph we have actually produced. In this case, we can get an algorithm that still uses $O(d \log n)$ copies like Theorem \ref{thm:lowdegreegs}, but is also computationally efficient, in that its runtime is $O(d^3 n \log^3 n)$.

\begin{thm}
\label{thm:boundeddegsubgraph}
Let $G'$ be a graph of bounded degree $d$, $G$ be a subgraph of $G'$. Given access to copies of $\ket{G}$, there is a quantum algorithm that identifies $G$ using $O(d \log n)$ copies with runtime $O(d^3 n \log^3 n)$. The algorithm succeeds with probability at least 0.99.
\end{thm}

\begin{proof}
We take $k$ Bell samples, for some $k$ to be determined. 
For each vertex $v$, the corresponding row $\r_v$ of $A$ is a linear combination over $\F_2$ of at most $d$ fixed vectors $\e_1,\dots,\e_d$ of Hamming weight 1, where each vector corresponds to a neighbour of $v$ in $G'$. 
So we can write $\r_v = \sum_{i=1}^d x_i \e_i$ for some $x_i \in \{0,1\}$, and determining $\x \in \{0,1\}^d$ is sufficient to determine $\r_v$. As the results of the Bell samples correspond to inner products between $\r_v$ and random vectors over $\F_2^n$, we obtain a system of $k$ random linear equations in $d$ unknowns. These equations can be solved in time $O(k^3)$ to determine $\x$ if the corresponding random matrix is full rank, and the probability that a random $k \times d$ matrix over $\F_2$, $k \ge d$, is not full rank is $O(2^{-(k-d)})$ \cite{ferreira2012rank}.
%\am{could cite ``The Rank of Random Binary Matrices and
%Distributed Storage Applications''?} 
So, by a union bound, it is sufficient to take $k = O(d \log n)$ for all of the rows of $A$ to be determined by solving the corresponding systems of linear equations.
\end{proof}

Using a similar technique to the last part of Theorem \ref{thm:lowdegreegs}, the linear dependence on $n$ in Theorem \ref{thm:boundeddegsubgraph} can be replaced with a linear dependence on the number of non-isolated vertices, and the polylogarithmic dependence on $n$ can be replaced with an equivalent dependence on $m$.

%\cs{It seems the technique in Theorem 13 can be used again to replace $n$ by $m$. Is this theorem about learning bounded graphs? Is $d$ the degree of $G$? Suppose we know some information of the subgroup, can we do better? For example, the graph is also bounded of degree $d'$.}

%-------------------------------------------------------------------------------

\section{Learning an unknown graph with parity queries}
\label{sec:parityqueries}

In this section we  investigate learning an unknown graph $G$ using the parity oracle $\fgpar(S)$. Identifying $S$ with a bit-string $x \in \{0,1\}^n$ via $x_i = 1$ if $i \in S$, and $x_i = 0$ otherwise, we see that
\[ \fgpar(x) = \sum_{(i,j) \in E} x_i x_j \]
where the sum is taken mod 2. So, if $G$ is arbitrary, $\fgpar$ is an arbitrary quadratic polynomial over $\F_2$ with no linear part.
It was shown in~\cite{montanaro12b} that any polynomial of this form can be learned using $O(n)$ quantum queries, and this is optimal. This immediately gives a quantum algorithm for learning an arbitrary graph using $O(n)$ parity queries, which is quadratically better than the best possible classical algorithm. (By an information-theoretic argument, classically $\Omega(n^2)$ parity queries are required.)

Evaluating $\fgpar(x)$ on a uniform superposition over computational basis states $\ket{x}$ gives precisely the graph state $\ket{G}$, so the results of Section \ref{sec:graphstates} can all immediately be applied to learning graphs in the parity query model. However, the ability to evaluate $\fgpar(x)$ on other input states allows for more general algorithms to be developed. In particular, we can obtain the following subroutine.

\begin{lem}
\label{lem:vector}
Let $A$ be the adjacency matrix of $G$. For any $v \in \{0,1\}^n$, there is a quantum algorithm which returns $Av$ and makes two queries to $\fgpar$.
\end{lem}

\begin{proof}
Consider the function $g_v(x) = f(x) + f(x + v)$. It can be evaluated for any $x$ using two queries to $f$. Let $B$ denote the adjacency matrix $A$, except that we set $B_{ij} = 0$ for $i > j$. Then $f(x) = x^T B x$.

We evaluate $g_v$ in superposition to produce
\beas
\frac{1}{\sqrt{2^n}} \sum_{x \in \{0,1\}^n} (-1)^{g_v(x)} \ket{x} &=& \frac{1}{\sqrt{2^n}} \sum_{x \in \{0,1\}^n} (-1)^{x^T B x + (x+v)^T B (x+v)} \ket{x}\\
&=& \frac{1}{\sqrt{2^n}} (-1)^{v^T B v} \sum_{x \in \{0,1\}^n} (-1)^{v^T B x + x^T Bv} \ket{x}\\
&=& \frac{1}{\sqrt{2^n}} (-1)^{v^T B v} \sum_{x \in \{0,1\}^n} (-1)^{x \cdot (Av)} \ket{x}.
\eeas
Then applying Hadamard gates to each qubit returns the vector $Av$.
\end{proof}

Note that no equivalent of Lemma \ref{lem:vector} can hold in the graph state model of Section \ref{sec:graphstates}. If we let $v$ be a vector of Hamming weight 1, Lemma \ref{lem:vector} returns an entire row of the adjacency matrix of $A$ using one query. But even to determine one entry of an arbitrary row of $A$ requires $\Omega(n)$ copies of $\ket{G}$, because this is equivalent to a quantum random access code on $\binom{n}{2}$ bits\footnote{Joe Fitzsimons, personal communication.}. Such codes are known to require quantum states of $\Omega(n^2)$ qubits~\cite{nayak99a}, and $\ket{G}$ is a state of $n$ qubits.

We can use Lemma \ref{lem:vector} as a subroutine to learn an arbitrary graph with a bounded number of edges. Classically, by an information-theoretic argument, this requires at least $\Omega(\log \binom{\binom{n}{2}}{m}) = \Omega(m \log(n^2/m) )$ queries.

\begin{thm}
\label{thm:boundednumedges_parity}
There is a quantum algorithm which learns a graph with at most $m$ edges using $O(\sqrt{m\log m})$ parity queries. The algorithm succeeds with probability at least 0.99.
\end{thm}

\begin{proof}
The algorithm splits the graph into low and high-degree parts. First, Theorem \ref{thm:lowdegreegs} is used with $d=\sqrt{m/\log m}$. This learns all rows of $A$ with at most $\sqrt{m/\log m}$ nonzero entries, and the identities of all ``dense'' rows of $A$ with more than $\sqrt{m/ \log m}$ nonzero entries. Then each of the dense rows is learned individually by applying Lemma \ref{lem:vector} with $v$ chosen to be the corresponding standard basis vector. There can be at most $O(\sqrt{m\log m})$ dense rows, so the overall algorithm uses $O(\sqrt{m\log m})$ queries.
\end{proof}

Theorem \ref{thm:boundednumedges_parity} is close to tight, because identifying an arbitrary graph on $k$ vertices (and hence with up to $\Theta(k^2)$ edges) requires $\Omega(k)$ quantum queries~\cite{montanaro12b}.
Stars and cliques can be learned with $O(1)$ parity queries via the techniques of the previous section for graph states.

%-------------------------------------------------------------------------------
%-------------------------------------------------------------------------------

\section{Combinatorial group testing}
\label{section:Combinatorial group testing}

Next, we move on from the problem of learning graphs to combinatorial group testing (CGT). 
In the CGT problem, we are given oracle access to an $n$-bit string $A$ with Hamming weight at most $k$.
Usually, we assume that $k\ll n$. 
In one query, we can get the OR of an arbitrary subset of the bits of $A$. The goal is to determine $A$ using the minimal number of queries. (To connect to the topic of the previous sections, we can see CGT as the problem of learning a graph on $n$ vertices with OR queries, in the very special case where the graph is promised to have no edges between vertices, and may contain up to $k$ self-loops.)

We can think of $A$ as a subset of $[n]$, and define
the oracle as
\be
f_A(S) = 
\begin{cases}
1, \quad {\rm if}~A \cap S \neq \emptyset, \\
0, \quad {\rm otherwise.}
\end{cases}
\ee
Classically, it is known that the number of queries required to solve CGT is $\Theta(k \log(n/k))$  \cite{du2000combinatorial}.
In the quantum case,
Ambainis and Montanaro \cite{ambainis14b} first studied this problem  and proposed a quantum algorithm using $O(k)$ queries. They also showed a lower bound of $\Omega(\sqrt{k})$.
Later in~\cite{belovs15}, based on the adversary bound method, Belovs proved that a quantum computer can solve the CGT problem with $\Theta(\sqrt{k})$ queries.
In principle, Belovs' approach can yield a quantum algorithm with an explicit implementation, but 
this implementation might not be time-efficient.
In this section, we  propose a quantum algorithm for CGT with an efficient implementation. 
The complexity is a little worse than $\Theta(\sqrt{k})$ by a factor of $O((\log k)(\log\log k))$.

The idea of our quantum algorithm is inspired by~\cite{ambainis16} and the Bernstein–Vazirani algorithm~\cite{bernstein97}. The key idea is to observe that the Bernstein-Vazirani algorithm allows the identity of a subset $A \subseteq [n]$ to be determined with one query to an oracle that computes $|A \cap T|$ for arbitrary $T \subseteq [n]$. And in~\cite{ambainis16}, Ambainis et al solved a closely related problem to evaluating this oracle, which they called gapped group testing (GGT):
given the oracle $f_A$, decide if $|A|\leq k$
or $|A|\geq k+d$. They showed that $\Theta(\sqrt{k/d})$ queries are enough to solve this problem by the adversary bound method.
The main idea of their quantum algorithm was borrowed from~\cite{belovs15}, but unlike~\cite{belovs15},
they have an efficient implementation
of their quantum algorithm.

So it seems that, by taking $d=1$ and using binary search, we can use
the quantum algorithm of~\cite{ambainis16}
for the gapped group testing problem to determine $|A|$ with $O(\log k)$ repetitions of their algorithm, leading to a query complexity of $O(\sqrt{k}\log k)$.
However, we should be careful at this point since the quantum algorithm of~\cite{ambainis16} only succeeds with probability $2/3$. So $O(\log k)$ repetitions will decrease the success probability to almost 0. A simple method to increase the success probability to $1-O((\log k)^{-1})$ is using the Chernoff bound. We can think of the intended output of the algorithm of~\cite{ambainis16} for GGT as 1 if $|A| \leq k$ and 0 if $|A|\geq k+1$. Denote this outcome $O$. As proved in~\cite{ambainis16}, the probability that each run of the algorithm returns the intended outcome is at least $2/3$. We repeat the algorithm for GGT $t$ times and output the median of the results. Let $X$ be the median. Then by the Chernoff bound, we have
$
{\rm Pr} [X \neq O] \leq e^{-c t}
$
for some constant $c$. So by choosing $t = O(\log\log k)$, the success probability is increased to $1-c'(\log k)^{-1}$ for an arbitrarily small constant $c'$. Taking a union bound over the $\lceil \log_2 k \rceil$ uses of the algorithm, we can determine $|A|$ with success probability $9/10$. By applying this algorithm to subsets $S \subseteq [n]$, for varying subsets $|S|$, we can determine $|A \cap S|$ with success probability $9/10$.

Next we show that access to an oracle of this form is sufficient to determine $A$ completely. In fact, this claim holds for any monotone function, rather than just the OR function.

\begin{lem}
\label{lem:monotonesize}
Consider a family of monotone boolean functions $g:\{0,1\}^k \to \{0,1\}$. Assume there is a family of classical or quantum algorithms $\mathcal{A}_n$ which, when applied to $f:\{0,1\}^n \to \{0,1\}$ such that $f(x) = g(x_S)$ for some subset $S$ such that $|S|=k$, outputs $k$ with success probability $9/10$. Let $T(n)$ denote the complexity of $\mathcal{A}_n$, and assume that $T(n)$ is nondecreasing. Then there is a quantum algorithm which determines $S$ with success probability $1-\delta$, for any $\delta > 0$, and has complexity $O(T(n) \log 1/\delta)$.
\end{lem}

\begin{proof}
Identify $n$-bit strings with subsets of $[n]$, and create the uniform superposition $\frac{1}{\sqrt{2^n}} \sum_{T \subseteq [n]} \ket{T}$. For each $T$, run $\mathcal{A}_{|T|}$ on the function $f_T:\{0,1\}^{|T|} \to \{0,1\}$ given by $f$ restricted to the variables in $T$. As $f$ is monotone, a query to $f_T$ can be simulated by a query to $f$ by setting the variables outside of $T$ to 0. The result is a state of the form
\[ \frac{1}{\sqrt{2^n}} \sum_{T \subseteq [n]} \ket{T} (\sqrt{1-\delta_T} \ket{|S \cap T|} + \sqrt{\delta_T} \ket{\psi_T}) \]
for some $\delta_T \in [0,1]$ such that $\delta_T \le 1/3$, and some states $\ket{\psi_T}$ such that $\braket{|S \cap T| | \psi_T} = 0$. Apply $Z^{\otimes |T|}$ to the last register and uncompute $\mathcal{A}_{|T|}$ to produce
\[ \frac{1}{\sqrt{2^n}} \sum_{T \subseteq [n]} (-1)^{|S \cap T|} (1-\delta_T) \ket{T} \ket{0} + \ket{\eta} \]
for some unnormalised state $\ket{\eta}$ orthogonal to $\ket{0}$ on the second register. Measure the second register and output ``fail'' if the result is not 0. Otherwise, apply Hadamard gates to every qubit of the remaining register, and return the result.

The algorithm outputs failure with probability $1 - \frac{1}{2^n} \sum_T (1-\delta_T)^2 \le 2\delta_T - \delta_T^2 \le 1/5$. If the algorithm does not output failure, the residual state has squared inner product $(\frac{1}{2^n} \sum_T (1-\delta_T))^2 \ge (9/10)^2$ with the state $\frac{1}{\sqrt{2^n}} \sum_{T \subseteq [n]} (-1)^{|S \cap T|} \ket{T}$; if applied to this state, it would output  $S$ with certainty, by the analysis of the Bernstein-Vazirani algorithm. Therefore the algorithm fails with probability at most $1/5 + 19/100 < 1/2$. Repetition and taking the majority vote reduces the failure probability to $\delta$, for arbitrary $\delta > 0$, with an additional multiplicative cost $O(\log 1/\delta)$.
\end{proof}

By Lemma \ref{lem:monotonesize}, we obtain the following theorem.

\begin{thm}
There is a quantum algorithm  that solves the CGT problem with success probability at least $2/3$. The
query complexity is $O(\sqrt{k}(\log k)(\log\log k))$, and
time complexity is $\widetilde{O}(n\sqrt{k})$.
\end{thm}

%-------------------------------------------------------------------------------

\section{Majority and exact-half functions}
\label{section:Majority and exact-half functions}

In this section, we consider the following general learning problem. We are given access to a function $f:\{0,1\}^n \to \{0,1\}$, which is promised to be equal to some known function $g:\{0,1\}^k \to \{0,1\}$ acting on a subset $S$ of the variables. Our goal is to learn which $k$ variables $f$ depends on. We first note that, for any function $g$, any classical algorithm for this problem must make $\Omega(\log \binom{n}{k}) = \Omega(k \log (n/k))$ queries, as each query returns 1 bit of information. For some functions $g$, quantum algorithms can do better. In particular, using the adversary bound method, Belovs~\cite{belovs15} showed that for the exact-half function ($g(x)=1 \Leftrightarrow |x|=k/2$) and the majority function ($g(x)=1 \Leftrightarrow |x| \ge k/2$), the quantum query complexity of identifying $S$ is $O(k^{1/4})$.
Here we give simple explicit quantum algorithms that match this complexity up to logarithmic factors. Then we observe that an even simpler approach can be used to solve this learning problem for almost all functions $g$.

The approach used in this section is based on applying Fourier sampling to $f$ (see Section \ref{sec:fs}), an approach explored by At{\i }c{\i } and Servedio~\cite{atici05} in the context of quantum learning and testing algorithms for functions with few relevant variables. Fourier sampling allows one to produce the state $\ket{\psi_f} = \sum_{T \subseteq [n]} \widehat{f}(T) \ket{T}$ with one quantum query to $f$. Now observe that, if $f(x)$ does not depend on the $i$'th bit $x_i$, $\widehat{f}(T) = 0$ for all $T$ such that $i \in T$. So, if $f$ depends only on a subset $S$ of the variables, measuring $\ket{\psi_f}$ in the computational basis returns a subset of $S$. By repeating this procedure we can hope to learn all of $S$, and we can sometimes accelerate this process using amplitude amplification. Let $W_l(g)$ be the Fourier weight of $g$ on the $l$'th level, $W_l(g) = \sum_{T,|T|=l} \widehat{g}(T)^2$. Similarly define $W_{\ge l}(g) = \sum_{T,|T|\ge l} \widehat{g}(T)^2$.

\begin{lem}
\label{lem:symalg}
Let $g$ be a symmetric function, i.e.\ $g(x) = h(|x|)$ for some $h$, where $|x|$ is the Hamming weight of $x$. Then, for any $l$ such that $W_{\ge l}(g) > 0$, there is a quantum algorithm which identifies $S$ with probability at least $0.99$ using $O(k/(l \sqrt{W_{\ge l}(g)})) \log k)$ queries to $f$. If $l=k$, there is a quantum algorithm using $O(1/\sqrt{W_k(g)})$ queries.
\end{lem}

\begin{proof}
We start by applying amplitude amplification~\cite{brassard02} to the following procedure: use Fourier sampling on $f$, and return ``yes'' if the size of the subset returned is at least $l$. This returns a subset $T$ of size $l' \ge l$ using $O(1/\sqrt{W_{\ge l}(g)})$ evaluations of $f$ and with success probability $\max\{1-W_{\ge l}(g), W_{\ge l}(g)\} \ge 1/2$~\cite[Theorem 2]{brassard02}. Observe that, as $\ket{\psi_f}$ has no support on subsets that are not contained within $S$, $T \subseteq S$ with certainty.

As $g$ is symmetric, $\widehat{f}(T)$ depends only on $|T|$ for all $T$, so $T$ is picked uniformly at random from all $l'$-subsets of $[k]$. For any $r$, it is sufficient to perform this procedure $O(r)$ times to achieve $r$ successes with high probability. The final step of the algorithm is to output the union of the subsets returned in successful iterations. By a union bound, the probability that there is a variable that is not included in any of the subsets is at most $k(1-l/k)^r \le k e^{-lr/k+r}$. So it is sufficient to take $r=O((k/l) \log (k/\delta))$ to achieve success probability $1-\delta$. For the second claim in the lemma, if $l=k$, we learn all the relevant variables with one use of amplitude amplification and with probability $\ge 1/2$, which can be boosted to arbitrarily close to 1 with a constant number of repetitions.
\end{proof}

Lemma \ref{lem:symalg} crucially relies on $g$ being symmetric. Otherwise, certain variables could be substantially harder to identify than others. To apply Lemma \ref{lem:symalg}, it is sufficient to find bounds on the Fourier spectrum of $g$, which we now obtain for certain functions. First, we consider the majority function (MAJ$_k(x)=1 \Leftrightarrow |x| \ge k/2$), which is a special case of a previously studied framework known as ``threshold group testing''~\cite{chen09}. 

\begin{fact}\cite[Theorem 3.5.3]{odonnell03}
\label{fact:maj}
Let {\rm MAJ}$_k$ be the majority function on $k$ bits. If $|S|$ is even, then $\widehat{{\rm MAJ}_k}(S) = 0$. Otherwise,
\[ \widehat{{\rm MAJ}_k}(S) = (-1)^{(k-1)/2} \frac{\binom{(k-1)/2}{(|S|-1)/2}}{\binom{k-1}{|S|-1}} \frac{2}{2^k} \binom{k-1}{(k-1)/2}. \]
%
%So, if $k$ is odd, $\widehat{\text{MAJ}_k}([k])^2 = \frac{4}{2^{2k}} \binom{k-1}{(k-1)/2}^2 = \Theta(1/k)$.
\end{fact}

Using Fact \ref{fact:maj}, we can obtain a bound on the tail of the Fourier spectrum of the majority function.

\begin{lem}
$W_{\ge (k+1)/2}({\rm MAJ}_k) = \Omega( 1/\sqrt{k})$.
\end{lem}

\begin{proof}
By Fact \ref{fact:maj},
\[ W_l(\text{MAJ}_k) = \binom{k}{l} \frac{\binom{(k-1)/2}{(l-1)/2}^2}{\binom{k-1}{l-1}^2} \frac{4}{2^{2k}} \binom{k-1}{(k-1)/2}^2 = \frac{k}{l} \frac{\binom{(k-1)/2}{(l-1)/2}^2}{\binom{k-1}{l-1}} \frac{4}{2^{2k}} \binom{k-1}{(k-1)/2}^2 \]
and using $\frac{4}{2^{2k}} \binom{k-1}{(k-1)/2}^2 = \Theta(1/k)$, we obtain
\[ W_l(\text{MAJ}_k) = \Theta\left( \frac{\binom{(k-1)/2}{(l-1)/2}^2}{k\binom{k-1}{l-1}}\right) \]
for $l \ge (k+1)/2$. In the case $l=(k+1)/2$, we have $W_l(\text{MAJ}_k) = \Theta(k^{-3/2})$ using $\binom{a}{a/2} = \Theta(2^a/\sqrt{a})$ for any $a$. By Stirling's formula, $\frac{\binom{(k-1)/2}{(l-1)/2}^2}{\binom{k-1}{l-1}} \approx \sqrt{\frac{2(k-1)}{\pi(l-1)(k-l)}}$, which is nondecreasing when $l \ge (k+1)/2$, so  $W_l(\text{MAJ}_k) = \Omega(k^{-3/2})$ for $l \ge (k+1)/2$.
\end{proof}

Next, we consider the EXACT-HALF function, $g(x) = 1 \Leftrightarrow |x| = k/2$.

\begin{lem}
Let $k$ be even. Then $W_{\ge k/2}({\rm EXACT}\text{-}{\rm HALF}_k) = \Theta(1/\sqrt{k})$.
\end{lem}

\begin{proof}
Let $g:\{0,1\}^k \to \{0,1\}$ be the EXACT-HALF function. It will be convenient for the proof to switch to the representation of the Fourier transform of $g$ that $\widehat{g}(s) = \frac{1}{2^k} \sum_{x \in \{0,1\}^k} (-1)^{s \cdot x} g(x)$, which is equivalent to the representation used in the rest of this paper for all $s$ such that $s \neq 0^k$, up to a constant factor. Then, for $s \neq 0^k$,
\[ \widehat{g}(s) = \sum_{x,|x|=k/2} (-1)^{x \cdot s} = \frac{1}{2^k} \sum_{i=0}^{k/2} (-1)^i \binom{|s|}{i} \binom{k-|s|}{i}, \]
where the last expression is a Krawtchouk polynomial~\cite{krasikov99}. This is symmetric about $|s|=k/2$, so
\[ 
\sum_{s,|s| \ge k/2} \widehat{g}(s)^2 \ge \frac{1}{2} \sum_s \widehat{g}(s)^2 = \frac{1}{2} \|g\|_2^2 = \Theta(1/\sqrt{k}). 
\qedhere
\]
\end{proof}

So, by the above lemmas, we reproduce the $\Theta(k^{1/4})$ complexity of Belovs' algorithms for the majority and EXACT-HALF functions up to a logarithmic factor. The algorithms are also time-efficient.

\begin{thm}
There exist quantum algorithms that learn the majority and exact-half functions on $k$-bits using $O(k^{1/4}\log k)$ queries. The time complexity is $O(nk^{1/4}\log k)$.
\end{thm}

%More generally, we have $\hat{f}(S) = \pm \hat{f}(S^c)$ if $f(x) = \pm (-1)^{|x|}f(x)$ for all $x$. This holds, for example, if $f(x) = 0$ for all $x$ such that $|x|$ is even; or if there is only one $x$ such that $f(x)=1$.

Finally, we observe a simple general approach which can be used to solve the learning problem for almost all functions efficiently.
Define the influence of the $j$'th variable as
\[ \Inf_j(g) = \sum_{T \ni j} \widehat{g}(T)^2 = \Pr_{x \in \{0,1\}^k} [g(x) \neq g(x^j)], \]
where $x^j$ is the bit-string equal to $x$ with its $j$'th bit flipped.

\begin{prop}[essentially At{\i }c{\i } and Servedio~\cite{atici05}]
\label{lem:highinf}
Assume that, for all $j \in S$, $\Inf_j(g) \ge \epsilon$. Then there is a quantum algorithm which identifies $S$ with probability $1-\delta$ using $O(\epsilon^{-1} \log (k/\delta))$ queries to $f$.
\end{prop}

\begin{proof}
We apply Fourier sampling to $f$, which returns a subset $T \subseteq [k]$ with probability $\widehat{g}(T)^2$. We use this subroutine $q$ times and output the union of the subsets of variables returned. The probability that the $j$'th variable is included in each sample is $\Inf_j(g) \ge \epsilon$. The probability that there exists a variable that is not returned after the $q$ queries is at most $k(1-\epsilon)^q \le ke^{-q\epsilon}$. So it is sufficient to take $q = O(\epsilon^{-1} \log (k/\delta))$ to learn all the variables except with probability $\delta$.
\end{proof}

If $g$ is picked at random, then for all $j$, $\Inf_j(g)$ is lower-bounded by a constant with high probability. So, by Proposition \ref{lem:highinf}, for almost all functions $g$, there is a quantum algorithm that identifies $S$ using $O(\log k)$ queries and succeeds with probability 0.99. This holds even if $g$ is unknown, and is an exponential improvement over the optimal classical complexity.

%-------------------------------------------------------------------------------

\section{Outlook}

We have seen that quantum algorithms can achieve relatively large speedups over their classical counterparts for learning unknown graphs in a variety of different models. Many of the results we obtained are tight up to logarithmic factors, but some larger gaps remain.
In particular, for learning Hamiltonian cycles and matchings using OR queries, the quantum lower bound is $\Omega(\sqrt{m})$~\cite{ambainis14b}, while our upper bounds are $\widetilde{O}(m^{3/4})$.

%-------------------------------------------------------------------------------

\section*{Acknowledgments}

We would like to thank Jo\~{a}o Doriguello and Ryan Mann for helpful discussions on the topic of this work. We acknowledge support from the QuantERA ERA-NET Cofund in Quantum Technologies implemented within the European Union’s Horizon 2020 Programme (QuantAlgo project) and EPSRC grants EP/R043957/1 and EP/T001062/1. This project has received funding from the European Research Council (ERC) under the European Union’s Horizon 2020 research and innovation programme (grant agreement No.\ 817581).

%-------------------------------------------------------------------------------

\bibliographystyle{plain}
\bibliography{main}

\end{document}